\documentclass{elsarticle}

    \usepackage{amsmath}
    \usepackage{amsthm}
	\usepackage{tikz}

    \newtheorem{theorem}{Theorem}
    \newtheorem{lemma}[theorem]{Lemma}
    \newtheorem{corollary}[theorem]{Corollary}
    \newtheorem{conjecture}[theorem]{Conjecture}
    
    \tikzstyle{vertex}=[circle,fill=black!100,minimum size=6pt,inner sep=0pt]
    \tikzstyle{edge} = [draw,thick,-]
    \tikzstyle{dashed-edge} = [draw,thick,-, dashed]
    \tikzstyle{very-thick-edge} = [draw,line width=0.8mm,-]
    \tikzstyle{dotted-edge} = [draw,thick,-, dotted]
    \tikzstyle{possible-path} = [edge, loosely dotted]
    \tikzstyle{very-thick-possible-path} = [very-thick-edge, dotted]
    \tikzstyle{color-label}=[font=\small,inner sep=0pt]
    \tikzstyle{graph-label}=[font=\small,inner sep=0pt]
    
\begin{document}
    \begin{frontmatter}
        \title{Edge coloring of graphs of signed class $1$ and $2$}
        \author{Robert Janczewski\fnref{janczewski_footnote}}
        \address{Department of Algorithms and System Modeling, Faculty of Electronics, Telecommunication and Informatics, Gdańsk University of Technology, ul.\ Narutowicza $11/12$, Gdańsk, Poland}
        \fntext[janczewski_footnote]{E-mail: skalar@eti.pg.gda.pl}
        
        \author{Krzysztof Turowski\corref{grant}\fnref{turowski_footnote}}
        \address{Theoretical Computer Science Department, Jagiellonian University, Kraków, $30$–$348$, Poland}
        \fntext[turowski_footnote]{E-mail: krzysztof.szymon.turowski@gmail.com}
        \cortext[grant]{Krzysztof Turowski's research was funded in whole by Polish National Science Center $2020$/$39$/D/ST$6$/$00419$ grant. For the purpose of Open Access, the author has applied a CC-BY public copyright licence to any Author Accepted Manuscript (AAM) version arising from this submission.}
        
        \author{Bartłomiej Wróblewski\corref{corresponding}\fnref{wroblewski_footnote}}
        \address{Department of Algorithms and System Modeling, Faculty of Electronics, Telecommunication and Informatics, Gdańsk University of Technology, ul.\ Narutowicza $11/12$, Gdańsk, Poland}
        \cortext[corresponding]{Corresponding author}
        \fntext[wroblewski_footnote]{E-mail: bart.wroblew@gmail.com}

        \begin{abstract}
            Recently, Behr \cite{behr} introduced a notion of the chromatic index of signed graphs and proved that for every signed graph $(G$, $\sigma)$ it holds that
            \begin{displaymath}
                \Delta(G)\leq\chi'(G\text{, }\sigma)\leq\Delta(G)+1\text{,}
            \end{displaymath}
            where $\Delta(G)$ is the maximum degree of $G$ and $\chi'$ denotes its chromatic index.
            
            In general, the chromatic index of $(G$, $\sigma)$ depends on both the underlying graph $G$ and the signature $\sigma$. In the paper we study graphs $G$ for which $\chi'(G$, $\sigma)$ does not depend on $\sigma$. To this aim we introduce two new classes of graphs, namely $1^\pm$ and $2^\pm$, such that graph $G$ is of class
            $1^\pm$ (respectively, $2^\pm$) if and only if $\chi'(G$, $\sigma)=\Delta(G)$ (respectively, $\chi'(G$, $\sigma)=\Delta(G)+1$) for all possible signatures $\sigma$.
            
            We prove that all wheels, necklaces, complete bipartite graphs $K_{r,t}$ with $r\neq t$ and almost all cacti graphs are of class $1^\pm$. Moreover, we give sufficient and necessary conditions for a graph to be of class $2^\pm$, i.e. we show that these graphs must have odd maximum degree and give examples of such graphs with
            arbitrary odd maximum degree bigger than $1$. 
        \end{abstract}
    \end{frontmatter}

\section{Introduction}
    In this paper we only consider simple, finite and undirected graphs. Graph $G$ has a set of vertices $V(G)$ and a set of edges $E(G)$ with $n(G)$, $m(G)$ denoting their cardinalities, respectively. By $\deg_G(v)$ we denote degree of a vertex $v$ in a graph $G$ and by $\Delta(G)$ the maximum degree of all vertices of $G$.
    \emph{Incidence} is a pair $(v$, $e)$, where $v$ is a vertex, $e$ is an edge and $v$ is one of the endpoints of $e$. Incidence $(v$, $e)$ will be shortly denoted by $v\colon e$ and the set of all graph's incidences will be denoted by $I(G)$. All other definitions and symbols are consistent with those defined in Diestel \cite{diestel}.
    
    Signed graphs were introduced in $1950$s by Harary \cite{harary_signed} as a generalization of simple graphs. Their main purpose was to better model social relations of disliking, indifference and liking. A \emph{signed graph} is a pair $(G$, $\sigma)$ where $G$ is a graph and $\sigma\colon E(G)\to\{\pm 1\}$ is a function. $G$ and
    $\sigma$ are called the \emph{underlying} graph of $(G$, $\sigma)$ and the \emph{signature} of $(G$, $\sigma)$, respectively. Edge $e\in E(G)$ will be called \emph{positive} (respectively, \emph{negative}) if and only if $\sigma(e)=1$ (respectively, $\sigma(e)=-1$). Cycle in $(G$, $\sigma)$ is called \emph{positive} (respectively, \emph{negative}) if product of signs of its edges is positive (respectively, negative). A signed graph with only positive cycles is \emph{balanced} otherwise it is \emph{unbalanced}.

    Switching of a set $V'\subseteq V(G)$ of a signed graph $(G$, $\sigma)$ is an operation resulting in a signed graph $(G$, $\sigma')$ such that for every edge $uv\in E(G)$ we have
    \begin{displaymath}
        \sigma'(uv)=
        \begin{cases}
            -\sigma(uv)\text{, } & \text{if exactly one of $u$, $v$ belongs to $V'$,} \\
            \sigma(uv)\text{, } & \text{otherwise.}
        \end{cases}
    \end{displaymath}
    For example, switching a vertex, i.e.\ switching a subset of vertices of cardinality one, negates the signs of its incident edges. If a signed graph $(G$, $\sigma')$ can be obtained by switching some of the vertices of $(G$, $\sigma)$, we say that $(G$, $\sigma')$ and $(G$, $\sigma)$ are \emph{switching equivalent}. It is
    well-known that switching equivalence is an equivalence relation in the set of all signed graphs with a fixed underlying graph.

    In $2020$, Behr \cite{behr} introduced a concept of edge coloring of signed graphs as a generalization of ordinary graph edge coloring. Let $n$ be a positive integer and
    \begin{displaymath}
        M_n=
        \begin{cases}
            \{0\text{, }\pm 1\text{, }\ldots\text{, }\pm k\}\text{, } & \text{if $n=2k+1$,}\\
            \{\pm 1\text{, }\ldots\text{, }\pm k\}\text{, } & \text{if $n=2k$.}
        \end{cases}
    \end{displaymath}
    An \emph{$n$-edge-coloring} of a signed graph $(G$, $\sigma)$ is a function $f\colon I(G)\to M_n$ such that $f(u\colon uv)=-\sigma(uv)f(v\colon uv)$ for each edge $uv\in E(G)$ and $f(u\colon e_1)\neq f(u\colon e_2)$ for all edges $e_1\neq e_2$ such that $u\colon e_1$, $u\colon e_2\in I(G)$. By $\chi'(G$, $\sigma)$ we denote
    \emph{the chromatic index} of a signed graph $(G$, $\sigma)$, i.e.\ the smallest $n$ for which $(G$, $\sigma)$ has an $n$-edge-coloring.

    Behr \cite{behr} proved that every signed path can be colored using $2$ colors and a signed cycle can be colored with $2$ colors if and only if it's balanced. The main result of Behr's article \cite{behr} is the generalized Vizing's theorem. We will call it the Behr's theorem to emphasize the importance of Behr's achievement.
    \begin{theorem}[Behr \cite{behr}]
        $\Delta(G)\leq\chi'(G$, $\sigma)\leq\Delta(G)+1$ for all signed graphs $(G$, $\sigma)$. \qed
    \end{theorem}

    In this paper we introduce two new classes of graphs, namely $1^\pm$ and $2^\pm$, such that graph $G$ is of class $1^\pm$ (respectively, $2^\pm$) if and only if $\chi'(G$, $\sigma)=\Delta(G)$ (respectively, $\chi'(G$, $\sigma)=\Delta(G)+1$) for all possible signatures $\sigma$. In Section \ref{section2} we relate these classes to previously known notions. We
    also show sufficient and necessary conditions for a regular signed graph $(G$, $\sigma)$ to be $\Delta(G)$-edge-colorable and use these results to obtain sufficient and necessary conditions for a regular signed graph to be of class $2^\pm$. Next, we show that graphs of class $2^\pm$ must be of odd maximum degree and there
    are such graphs for any greater than $1$ odd value of the degree. In Section \ref{section3} we prove that all cacti except cycles are of class $1^\pm$. Section \ref{section4} deals with wheels---they are all of class $1^\pm$. In Section \ref{section5} we show that all necklaces except cycles are of class $1^\pm$.
    In Section \ref{section6} we deal with complete bipartite graphs by showing that graphs $K_{r,t}$ are of class $1^\pm$ for any $r\neq t$. The paper ends with some open problems and conjectures.

\section{$1^\pm$ and $2^\pm$}\label{section2}
    Behr \cite{behr_arxiv} defined \emph{class ratio $\mathcal{C}(G)$} of graph $G$ to be the number of possible signatures $\sigma\colon E(G)\to\{\pm 1\}$ such that the signed graph $(G$, $\sigma)$ is $\Delta(G)$-colorable, divided by the number of all possible signatures defined on $E(G)$, i.e.\ $2^{m(G)}$. The class ratio is a
    rational number satisfying $0\leq\mathcal{C}(G)\leq 1$ and, by previous definitions, $G$ is of class $1^\pm$ (respectively, $2^\pm$) if and only if $\mathcal{C}(G)=1$ (respectively, $\mathcal{C}(G)=0$). 
    Clearly, there are graphs that are neither of class $1^\pm$ nor $2^\pm$. The obvious example is a cycle---all signed cycles are $2$-colorable if are balanced and
    require $3$ colors otherwise.
    It is also easy to see that if $G$ is of class $1^\pm$ ($2^\pm$), then it is of class $1$ ($2$), where class $1$ ($2$) is defined as all graphs $G$ whose
    ordinary chromatic index equals $\Delta(G)$ ($\Delta(G)+1$). This relation provides justification for our notation: class $X^\pm$ is both a signed counterpart of class $X$ and a subclass of $X$.

    Recall that graph $G$ is \emph{$r$-regular} if and only if $\deg_G(v)=r$ for all $v\in V(G)$. Recall also that $M\subseteq E(G)$ is a (\emph{perfect}) \emph{matching} in $G$ if every vertex of $G$ is incident to at most (exactly) one edge of $M$.
    \begin{lemma}\label{regular_lemma}
        Let $G$ be a $k$-regular graph and $\sigma\colon E(G)\to\{\pm 1\}$ be a signature. $\chi'(G$, $\sigma)=\Delta(G)$ if and only if one of the following conditions holds:
        \begin{enumerate}
            \item $k=2r$ where $r$ is a positive integer and $(G$, $\sigma)$ admits a decomposition into exactly $r$ spanning edge disjoint $2$-regular balanced subgraphs.
            \item $k=2r+1$ where $r$ is a positive integer and $(G$, $\sigma)$ admits a decomposition into a perfect matching and exactly $r$ spanning edge disjoint $2$-regular balanced subgraphs.
        \end{enumerate}
    \end{lemma}
    \begin{proof}
        We prove the equivalence by proving both implications.
        
        ($\Leftarrow)$ In the $k=2r$ case, each of the $r$ spanning edge disjoint $2$-regular balanced subgraphs can be colored using exactly $2$ colors, so in total $2r$ colors are sufficient to color the whole graph. In the $k=2r+1$ case, the matching can be colored using color $0$, each of the $r$ spanning edge disjoint $2$-regular
        balanced subgraphs can be colored using $2$ colors, so in total $2r+1$ colors are also sufficient. Since $\Delta(G)=2r$ in the first case and $\Delta(G)=2r+1$ in the second, we get $\chi'(G$, $\sigma)=\Delta(G)$.
        
        ($\Rightarrow)$ We consider the two cases separately:
        \begin{enumerate}
            \item[$(1)$] $G$ is a $2r$-regular graph, so $m(G)=r\cdot n(G)$. Let $c$ be an arbitrary $\Delta(G)$-edge-coloring of $(G$, $\sigma)$. For every pair of opposite colors used by $c$, edges colored with them form a subgraph with maximum degree not exceeding $2$. Clearly, all such subgraphs are edge disjoint. Each of the
            subgraphs has at most $n(G)$ edges. There are exactly $r$ such subgraphs in $G$ and $G$ has $r\cdot n(G)$ edges, so every subgraph contains exactly $n(G)$ edges. If a subgraph with maximum degree not exceeding $2$ contains $n(G)$ edges, it clearly is spanning and $2$-regular. If it can be colored with two colors, it has to
            be balanced, too.
            
            \item[$(2)$] $G$ is a $(2r+1)$-regular graph, so $m(G)=r\cdot n(G)+\frac{1}{2}n(G)$. Let $c$ be an arbitrary $\Delta(G)$-edge-coloring of $(G$, $\sigma)$. The set of edges colored in $c$ with color $0$ must form a matching in $G$. Matching in $G$ can have maximum size of $\frac{1}{2}n(G)$, so there have to be at least
            $r\cdot n(G)$ edges colored by $c$ with colors different than $0$. Such edges form $r$ subgraphs with maximum degree not exceeding $2$, each of them is colored with $\{-k, k\}$ for some $k \neq 0$. Clearly, they are edge disjoint subgraphs with degree not exceeding $2$, so they have in total at most $r\cdot n(G)$
            edges. Combining these observations, we conclude that these subgraphs have exactly $r\cdot n(G)$ edges. This implies that they are $2$-regular, spanning and balanced -- since otherwise they could not be colored using only colors $\{-k, k\}$. This also implies that our matching has $\frac{1}{2}n(G)$ edges, which means it is perfect.
        \end{enumerate}
    \end{proof}
    Let $\mathbf{1}_X$ denote the function $X\ni x\mapsto 1\in\{\pm 1\}$.
    \begin{corollary}\label{even_delta_not_class_2}
        If $G$ is a graph with even $\Delta(G)$, then $G$ is not of class $2^\pm$.
    \end{corollary}
    \begin{proof}
        Every graph has a regular supergraph with the same maximum degree. Let $H$ be a $\Delta(G)$-regular supergraph of $G$. It follows from Petersen's $2$-factor theorem \cite{mulder1992julius} that $H$ admits a decomposition into $r$ edge disjoint $2$-regular, spanning subgraphs. Since all edges in a signed graph $(H$, $\mathbf{1}_{E(H)})$ are positive,
        these subgraphs are balanced. Lemma \ref{regular_lemma} implies that $\chi'(H$, $\mathbf{1}_{E(H)})=\Delta(H)$. Hence $\chi'(G$, $\mathbf{1}_{E(G)})\leq\chi'(H$, $\mathbf{1}_{E(H)})=\Delta(H)=\Delta(G)$, which shows that $G$ is not of class $2^\pm$.
    \end{proof}
    \begin{theorem}\label{class_2_conditions}
        Let $G$ be a graph. $G$ is of class $2^\pm$ if and only if $\Delta(G)$ is odd and there is no matching $M$ in $G$ such that $\Delta(G\setminus M)<\Delta(G)$.
    \end{theorem}
    \begin{proof}
        If follows from Corollary \ref{even_delta_not_class_2} that if graph is of class $2^\pm$ then its maximum degree must be odd.
        
        ($\Rightarrow)$ Let $G$ be a graph with odd maximum degree. Suppose that there is a matching $M\subseteq E(G)$ such that $\Delta(G\setminus M)<\Delta(G)$. We will show that $G$ is not of class $2^\pm$. Clearly, $\Delta(G\setminus M)$ is even since $\Delta(G\setminus M)=\Delta(G)-1$. The proof of Corollary 
        \ref{even_delta_not_class_2} shows that a signed graph $(G\setminus M$, $\mathbf{1}_{E(G)\setminus M})$ is $\Delta(G\setminus M)$-colorable. Since $\Delta(G\setminus M)$ is even, $\Delta(G\setminus M)$-colorings of $G\setminus M$ does not use color $0$. We can extend any of these colorings to a coloring of $G$ by using color $0$
        on $M$, so $(G$, $\mathbf{1}_{E(G)})$ is $\Delta(G)$-edge-colorable. This completes the proof.
        
        ($\Leftarrow$) Let us assume $\Delta(G)$ is odd, there is no matching $M \subseteq E(G)$ such that $\Delta(G\setminus M)<\Delta(G)$ and $G$ is not of class $2^\pm$. We show that this leads to a contradiction. There must exist a signature $\sigma$ such that $(G$, $\sigma)$ is $\Delta(G)$-edge-colorable. Let $c$ be an arbitrary
        $\Delta(G)$-edge-coloring of $(G$, $\sigma)$. Edges colored in $c$ with color $0$ must form a matching $M$. Clearly, $\Delta(G\setminus M)=\Delta(G)$, so graph $(G\setminus M$, $\sigma|_{E(G)\setminus M})$ requires at least $\Delta(G)$ colors different than $0$. It follows that $(G$, $\sigma)$ cannot be colored using $\Delta(G)$
        colors, so $G$ is of class $2^\pm$.
    \end{proof}
    It's easy to observe that signed graphs with $\Delta=1$ can be colored using color $0$ regardless of their signature, so graphs with $\Delta=1$ are of class $1^\pm$.
    \begin{theorem}
        For every $k \geq 1$ there is a graph $G$ of class $2^\pm$ such that $\Delta(G)=2k+1$.
    \end{theorem}
    \begin{proof}
        We will describe the procedure of construction of a graph $G$ such that $\Delta(G)=2k+1$ and $G$ is of class $2^\pm$. We start by constructing a complete bipartite graph $H$ with parts $\{u_1$, \ldots, $u_{2k}\}$ and $\{v_1$, \ldots, $v_{2k}\}$. Next we create a graph $H'$ by adding edges $u_{2i-1}u_{2i}$ for
        $1\leq i\leq k$ to $H$, a new vertex $v$ and edges $v_iv$ for $1\leq i\leq 2k$. It is easy to see that 
        \begin{displaymath}
            \deg_{H'}(w)=
            \begin{cases}
                2k\text{, } & \text{if $v=w$,}\\
                2k+1\text{, } & \text{otherwise.}
            \end{cases}
        \end{displaymath}
        Now we construct the graph $G$ by creating two disjoint copies of $H'$ and connecting their lowest degree vertices with path $P_3$ in such a way that these vertices are the endpoints of the path (see Figure \ref{fig3} for an example). By $v'_1$, $v_c$, $v'_2$ we denote vertices of that path and by $v'_1v_c$, $v_cv'_2$---its
        edges. Let's observe that $\deg_G(v_c)=2$ and other vertices of $G$ have degree $2k+1$.
        
        To complete the proof it suffices to show that $G$ is of class $2^\pm$. Clearly, $\Delta(G)$ is odd, so Theorem \ref{class_2_conditions} tells us that we need to show that there is no matching $M\subseteq E(G)$ such that $\Delta(G\setminus M)<\Delta(G)$. Let's assume there is such a matching $M$ in $G$. We show that it leads to
        a contradiction. Clearly, $M$ must cover all the vertices of $G$ other than $v_c$---it must cover all the vertices of both copies of $H'$. Let's observe that $H'$ has an odd number of vertices, so there is no perfect matching in it. It means that in order to cover vertices of both copies of $H'$, $M$ must contain edges
        $v'_1v_c$ and $v_cv'_2$. Then $M$ contains two edges incident with $v_c$, so $M$ is not a matching, a contradiction that completes the proof.
    \end{proof}
    \begin{figure}
        \centering
        \begin{tikzpicture}[thick]
            \node[vertex] (u1) at (-0.5,0) {};
            \node[vertex] (u2) at (-0.5,1) {};
            \node[vertex] (u3) at (-0.5,2) {};
            \node[vertex] (u4) at (-0.5,3) {};
            \node[vertex] (v1) at (1,0) {};
            \node[vertex] (v2) at (1,1) {};
            \node[vertex] (v3) at (1,2) {};
            \node[vertex] (v4) at (1,3) {};
            \node[vertex] (v) at (2, 1.5) {};
            \node[vertex] (center) at (3, 1.5) {};
            \node[vertex] (v') at (4, 1.5) {};
            \node[vertex] (u1') at (6.5,0) {};
            \node[vertex] (u2') at (6.5,1) {};
            \node[vertex] (u3') at (6.5,2) {};
            \node[vertex] (u4') at (6.5,3) {};
            \node[vertex] (v1') at (5,0) {};
            \node[vertex] (v2') at (5,1) {};
            \node[vertex] (v3') at (5,2) {};
            \node[vertex] (v4') at (5,3) {};
        
            \foreach \source / \dest in
                {u1/v1,u1/v2,u1/v3,u1/v4,
                u2/v1,u2/v2,u2/v3,u2/v4,
                u3/v1,u3/v2,u3/v3,u3/v4,
                u4/v1,u4/v2,u4/v3,u4/v4,
                v1/v,v2/v,v3/v,v4/v,
                v/center,center/v',
                v1'/v',v2'/v',v3'/v',v4'/v',
                u1'/v1',u1'/v2',u1'/v3',u1'/v4',
                u2'/v1',u2'/v2',u2'/v3',u2'/v4',
                u3'/v1',u3'/v2',u3'/v3',u3'/v4',
                u4'/v1',u4'/v2',u4'/v3',u4'/v4'}
                \path [edge] (\source) -- (\dest);
            \foreach \source / \dest in
                {u1/u2,u3/u4}
                \path (\source) edge [bend left=70] (\dest);
            \foreach \source / \dest in
                {u1'/u2',u3'/u4'}
                \path (\source) edge [bend right=70] (\dest);
        \end{tikzpicture}
        \caption{Example graph of class $2^\pm$ with $\Delta = 5$.\label{fig3}}
    \end{figure}
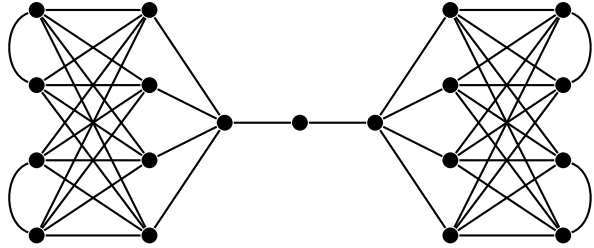

\section{Cacti}\label{section3}
    A connected graph $G$ is called \emph{cactus} if and only if every edge of $G$ belongs to at most one cycle. Clearly every tree is a cactus as it doesn't have cycles at all. It is a well known fact \cite{janczewski_cactus} that graph $G$ is a cactus if and only if there exists a sequence of graphs $G_1$, \ldots, $G_k$ called a
    \emph{decomposition of $G$}, such that:
    \begin{enumerate}
        \item $G_i$ is a cycle or a path $P_2$ for $1\leq i\leq k$;
        \item $G_i$ has exactly one vertex in common with $G_1\cup\ldots\cup G_{i-1}$ for $2\leq i\leq k$;
        \item $G = G_1\cup\ldots\cup G_k$. 
    \end{enumerate}
    We prove that almost all signed cacti are of class $1^\pm$.
    \begin{theorem}\label{cactus_coloring}
        Let $G$ be a cactus. If $G$ is not a cycle, then $G$ is of class $1^\pm$.
    \end{theorem}
    \begin{proof}
        Let $\sigma\colon E(G)\to\{\pm 1\}$ be a function. If $\Delta(G)\leq 2$, then $G$ is a path, so our claim is obvious. Therefore we assume that $\Delta(G)\geq 3$ in the remainder of the proof.
        
        Let $G_1$, \ldots, $G_k$ be an arbitrary decomposition of $G$ meeting the three above conditions. By $G'_l$ we denote $G_1\cup\ldots\cup G_l$. Clearly, $G'_k = G$. We will use induction on $l$ to prove that $(G'_l$, $\sigma|_{E(G'_l)})$ is $\Delta(G)$-edge-colorable. Let us observe that $G'_1$ is a cycle or path, so $(G'_1$,
        $\sigma|_{E(G'_1)})$ can be colored using $\Delta(G)\geq 3$ colors. Let's assume that $(G'_l$, $\sigma|_{E(G'_l)})$ is $\Delta(G)$-edge-colorable for any $l<x$. We will show that $(G'_x$, $\sigma|_{E(G'_x)})$ can be colored using at most $\Delta(G)$ colors. Let us note that $G'_x=G'_{x-1}\cup G_x$ and $(G'_{x-1}$,
        $\sigma|_{E(G'_{x-1})})$ is $\Delta(G)$-edge-colorable. By $c$ we denote an arbitrary $\Delta(G)$-edge-coloring of $(G'_{x-1}$, $\sigma|_{E(G'_{x-1})})$. It follows from the definition of the decomposition of $G$ that $G_x$ is either a path $P_2$ or a cycle and $G_x$ has exactly one vertex in common with $G'_{x-1}$. Let us
        consider the two cases separately.
        \begin{enumerate}
            \item $G_x$ is a path $P_2$. By $u$, $v$ we denote vertices of $G_x$ and without loss of generality we assume $u \in V(G'_{x-1})$. Clearly, there are at most $\Delta(G)-1$ edges incident to $u$ in $G'_{x-1}$, so there must be a color $\alpha\in M_{\Delta(G)}$ not used on any of these edges by $c$. Let's color $u\colon uv$
            with $\alpha$ and $v\colon uv$ with $-\sigma(uv)\alpha$. That way $(G'_x$, $\sigma|_{E(G'_x)})$ can be colored using $\Delta(G)$ colors.
            
            \item\label{cacti_case_cycles} $G_x$ is a cycle. By $u\in V(G_x)$ we denote a vertex such that $u\in V(G'_{x-1})$. By $uv_1$, $uv_2$ we denote two different edges of $G_x$. Clearly, there are at most $\Delta(G)-2$ edges incident to $u$ in $G'_{x-1}$, so there must be two different colors $\alpha$, $\beta\in M_{\Delta(G)}$
                not used on any of these edges by $c$. Let us consider three possible subcases.
                \begin{enumerate}
                    \item\label{cacti_cycle_1} $\alpha=-\beta$ and $\Delta(G)=3$. Since $\Delta(G)=3$, $0\in M_{\Delta(G)}$ and $\alpha\neq 0$, $\beta\neq 0$. By $w$ we denote vertex in $G_x$ such that $w$ is adjacent to $v_1$ and $w\neq u$. Incidences $v_1\colon v_1w$ and $w\colon v_1 w$ of $G'_x$ can be colored with color $0$. Let
                        us observe that other edges of $G_x$ form a path, so their incidences can be colored with colors $\pm\alpha$ in the coloring of $(G'_x$, $\sigma|_{E(G'_x)})$.
                    
                    \item\label{cacti_cycle_2} $\alpha=-\beta$ and $\Delta(G)>3$. Since $\Delta(G)>3$, there must be colors $\gamma$, $-\gamma\in M_{\Delta(G)}$ such that $\gamma\neq 0$, $\gamma\neq\alpha$, $\gamma\neq\beta$. Let us observe that edges $v_1u$, $uv_2$ form a path in $G_x$, so their incidences can be colored with colors
                        $\pm\alpha$ in the coloring of $(G'_x$, $\sigma|_{E(G'_x)})$. Other edges of $G_x$ also form a path---can be colored with colors $\pm\gamma$.
                    
                    \item\label{cacti_cycle_3} $\alpha\neq -\beta$. Without loss of generality let us assume $\beta\neq 0$. Incidence $u\colon uv_1$ can be colored with $\alpha$ in the coloring of $(G'_x$, $\sigma|_{E(G'_x)})$ and $v_1\colon uv_1$---with $-\sigma(uv_1)\alpha$. Edges of $G_x$ other than $uv_1$ form a path and can be
                        colored with $\pm\beta$ in such a way that $u\colon uv_2$ gets color $\beta$.
                \end{enumerate}
        \end{enumerate}
    \end{proof}
\begin{figure}
    \centering
    \begin{tikzpicture}[thick, scale=0.7]
        \node[vertex, label={[xshift=-1em, yshift=-1em] $u$}] (1) at (0,0) {};
        \node[vertex, label={[xshift=0em, yshift=0em] $v_1$}] (2) at (1.5,1.5) {};
        \node[vertex, label={[xshift=0em, yshift=-2em] $v_2$}] (3) at (1.5,-1.5) {};
        \node[vertex, label={[xshift=1em, yshift=-1em] $w$}] (4) at (3.5, 0.75) {};
        \node[color-label] (uv1_label) at (0.1,0.5) {$\alpha$};
        \node[color-label] (v1u_label) at (0.9,1.4) {$\pm\alpha$};
        \node[color-label] (uv2_label) at (-0.1,-0.5) {$-\alpha$};
        \node[color-label] (v2u_label) at (0.9,-1.4) {$\pm\alpha$};
        \node[color-label] (v1w_label) at (2,1.7) {$0$};
        \node[color-label] (wv1_label) at (3.2,1.25) {$0$};
        \node[color-label] (path_label) at (4,-1) {$\alpha$, $-\alpha$};
        \node[graph-label] (label) at (1.75,-2.6) {Case \ref{cacti_cycle_1}};
    
        \foreach \source / \dest in {1/2,1/3,2/4}
                \path [edge] (\source) -- (\dest);
        \path [possible-path] (4) edge [bend left=70] node {} (3);
    \end{tikzpicture}
    \quad
    \begin{tikzpicture}[thick, scale=0.7]
        \node[vertex, label={[xshift=-1em, yshift=-1em] $u$}] (1) at (0,0) {};
        \node[vertex, label={[xshift=0em, yshift=0em] $v_1$}] (2) at (1.5,1.5) {};
        \node[vertex, label={[xshift=0em, yshift=-2em] $v_2$}] (3) at (1.5,-1.5) {};
        \node[color-label] (uv1_label) at (0.1,0.5) {$\alpha$};
        \node[color-label] (v1u_label) at (0.9,1.4) {$\pm\alpha$};
        \node[color-label] (uv2_label) at (-0.1,-0.5) {$-\alpha$};
        \node[color-label] (v2u_label) at (0.9,-1.4) {$\pm\alpha$};
        \node[color-label] (path_label) at (3.2,0) {$\gamma$, $-\gamma$};
        \node[graph-label] (label) at (1.75,-2.6) {Case \ref{cacti_cycle_2}};
    
        \foreach \source / \dest in {1/2,1/3}
                \path [edge] (\source) -- (\dest);
        \path [possible-path] (2) edge [bend left=80] node {} (3);
    \end{tikzpicture}
    \quad
    \begin{tikzpicture}[thick, scale=0.7]
        \node[vertex, label={[xshift=-1em, yshift=-1em] $u$}] (1) at (0,0) {};
        \node[vertex, label={[xshift=0em, yshift=0em] $v_1$}] (2) at (1.5,1.5) {};
        \node[vertex, label={[xshift=0em, yshift=-2em] $v_2$}] (3) at (1.5,-1.5) {};
        \node[color-label] (uv1_label) at (0.1,0.5) {$\alpha$};
        \node[color-label] (v1u_label) at (0.9,1.4) {$\pm\alpha$};
        \node[color-label] (uv2_label) at (0.1,-0.5) {$\beta$};
        \node[color-label] (v2u_label) at (0.9,-1.4) {$\pm\beta$};
        \node[color-label] (path_label) at (3.2,0) {$\beta$, $-\beta$};
        \node[graph-label] (label) at (1.75,-2.6) {Case \ref{cacti_cycle_3}};
    
        \foreach \source / \dest in {1/2,1/3}
                \path [edge] (\source) -- (\dest);
        \path [possible-path] (2) edge [bend left=80] node {} (3);
    \end{tikzpicture}
    \caption{Three possible edge colorings of $(G_x$, $\sigma|_{E(G'_x)})$ considered in the case \ref{cacti_case_cycles} of the proof of Theorem \ref{cactus_coloring}.}
\end{figure}
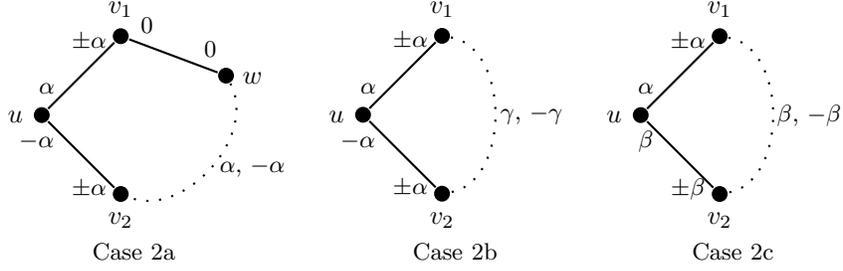

The decomposition of an $n$-vertex cactus can be done in $O(n)$ time. This implies that the coloring procedure being a part of the above proof is also linear.

\section{Wheels}\label{section4}
    Wheel $W_n$ is a graph on $n$ vertices consisting of a cycle $C_{n-1}$, a vertex $v\notin V(C_{n-1})$ and edges between $v$ and all the vertices of a cycle $C_{n-1}$.

    \begin{theorem}
        All wheels are of class $1^\pm$.
    \end{theorem}
    \begin{proof}
        We consider three cases separately:
        \begin{enumerate}
            \item $n = 4$. Wheel $W_4$ is a complete graph $K_4$. It's easy to observe that $W_4$ has a decomposition into three perfect matchings and edges of any two of them create a spanning cycle. Let us consider an arbitrary signed graph $(W_4$, $\sigma)$ and it's decomposition $D$ into three perfect matchings. If there is an odd
            number of negative edges in $(W_4$, $\sigma)$, there must be a matching in $D$ with an odd number of negative edges. Remaining edges create a balanced spanning cycle. In the opposite case---when there is an even number of negative edges in $(W_4$, $\sigma)$, there must be a matching in $D$ with an even number of negative
            edges. Remaining edges create a balanced spanning cycle. In both cases, $(W_4$, $\sigma)$ has a decomposition into a perfect matching and a balanced spanning cycle. Obviously, $W_4$ is $3$-regular, so it follows from Lemma \ref{regular_lemma} that $(W_4$, $\sigma)$ is $\Delta$-edge-colorable.
            
            \item $n = 2k + 1$, $k \geq 2$. Let $V(W_n) = \{u$, $v_0$, \ldots, $v_{n-2}\}$ and $E(G) = \{uv_0$, \ldots, $uv_{n-2}\} \cup \{v_0 v_1$, $v_1 v_2$, \ldots, $v_{n-2} v_0\}$. For $0 \leq i \leq k - 2$, let $G_i$ be a subgraph of $W_n$, such that $V(G_i) = \{v_i$, $v_{i+1}$, $u$, $v_{i+k}$, $v_{i+k+1}\}$ and $E(G_i) =    
                \{v_{i+1} v_i$, $v_i u$, $u v_{i+k}$, $v_{i+k} v_{i+k+1}\}$. Clearly, $G_i$ is a path and for any different $i_1$, $i_2$ paths $G_{i_1}$, $G_{i_2}$ are edge disjoint.
            
                Let $P$ be a subgraph of $W_n$, such that $V(P) = \{v_{k-1}$, $v_{k}$, $u$, $v_{n-2}$, $v_{0}\}$ and $E(P) = \{v_{k} v_{k-1}$, $v_{k-1} u$, $u v_{n-2}$, $v_{n-2} v_{0}\}$. Clearly, $P$ is a path. We will show that paths $P$, $G_i$ are edge disjoint for any $0 \leq i \leq k - 2$. It's sufficient to show that $\{v_i u$,
                $u v_{i+k}\} \cap \{v_{k-1} u$, $u v_{n-2}\} = \emptyset$ and $\{v_{i+1} v_i$, $v_{i+k} v_{i+k+1}\} \cap \{v_{k} v_{k-1}$, $v_{n-2} v_{0}\} = \emptyset$. The first equality follows from the fact that $i < k - 1 < i + k < n - 2$. Let's assume that the second equality is false, so one of the following cases must hold:
                \begin{enumerate}
                    \item $v_{i+1} v_i = v_k v_{k-1}$. Impossible because $i < k - 1 < k$.
                    \item $v_{i+1} v_i = v_{n-2} v_{0}$. Impossible because $i < i + 1 < n - 2$.
                    \item $v_{i+k} v_{i+k+1} = v_k v_{k-1}$. Impossible because $k - 1 < i + k < i + k + 1$.
                    \item $v_{i+k} v_{i+k+1} = v_{n-2} v_{0}$. Impossible because $i + k \neq 0$ and $i + k + 1 \neq 0$.
                \end{enumerate}
                Let's observe that $m(\bigcup\limits_{i=0}^{k - 2} G_{i}) + m(P) = 4(k - 1) + 4 = 4k = m(W_n)$, so $\bigcup\limits_{i=0}^{k - 2} G_{i} \cup P = W_n$ and $W_n$ has a decomposition into exactly $k = \Delta(W_n) / 2$ paths. It follows that any signed graph with $W_n$ as an underlying graph can be colored using
                $\Delta(W_n)$ colors because edges of each path can be colored with exactly two colors.
            
                \begin{figure}
                    \centering
                    \begin{tikzpicture}[thick, scale=0.7]
                        \node[vertex] (1) at (0,0) {};
                        \node[vertex] (2) at (-1.7,0.85) {};
                        \node[vertex] (3) at (0,2) {};
                        \node[vertex] (4) at (1.7,0.85) {};
                        \node[vertex] (5) at (1.7,-0.85) {};
                        \node[vertex] (6) at (0,-2) {};
                        \node[vertex] (7) at (-1.7,-0.85) {};
                    
                        \foreach \source / \dest in {2/3,3/1,1/6,6/5}
                                \path [dashed-edge] (\source) -- (\dest);
                        \foreach \source / \dest in {3/4,4/1,1/7,7/6}
                                \path [dotted-edge] (\source) -- (\dest);
                        \foreach \source / \dest in {7/2,2/1,1/5,5/4}
                                \path [edge] (\source) -- (\dest);
    
                    \end{tikzpicture}
                    \caption{The decomposition of $W_7$ into three paths. Distinct paths are marked with solid, dashed and dotted lines.}
                \end{figure}
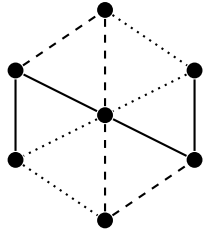
            
            \item $n = 2k$, $k \geq 3$. Let us consider wheel $W_{n-1}$. By $u$ we denote the center vertex of $W_{n-1}$---the one with degree $\Delta(W_{n-1})$ and by $v_1$, $v_2$ we denote arbitrary vertices adjacent to $u$ such that $v_1 v_2 \in E(W_{n-1})$. It follows from the previous case that $W_{n-1}$ admits a decomposition $D$
                into $k - 1$ paths. We construct graph $W_{n-1}'$ from $W_{n - 1}$ by adding new vertex $v$ and replacing edge $v_1 v_2$ by path $v_1$, $v$, $v_2$. It's easy to observe that $W_{n-1}'$ also admits some decomposition $D'$ into $k - 1$ paths. It can be constructed directly from $D$---path containing edge $v_1 v_2$ in $D$
                contains edges $v_1 v$, $v v_2$ in $D'$. Clearly, graph $W_n$ can be obtained from $W_{n-1}'$ by adding edge $v u$. It follows that $W_n$ admits a decomposition into $k - 1$ paths and a single edge, so every signed graph with $W_n$ as an underlying graph can be colored using $2k - 1 = \Delta(W_n)$ colors.
            
                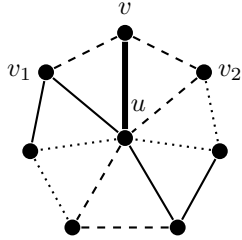
\begin{figure}
                    \centering
                    \begin{tikzpicture}[thick, scale=0.7]
                        \node[vertex, label={[xshift=0.5em, yshift=0.3em] $u$}] (1) at (0,0) {};
                        \node[vertex, label={[xshift=-1em, yshift=-1em] $v_{1}$}] (2) at (-1.5,1.25) {};
                        \node[vertex, label={[xshift=0em, yshift=0em] $v$}] (3) at (0,2) {};
                        \node[vertex, label={[xshift=1em, yshift=-1em] $v_{2}$}] (4) at (1.5,1.25) {};
                        \node[vertex] (5) at (1.8,-0.25) {};
                        \node[vertex] (6) at (1,-1.7) {};
                        \node[vertex] (7) at (-1,-1.7) {};
                        \node[vertex] (8) at (-1.8,-0.25) {};
                    
                        \foreach \source / \dest in {2/3,3/4,4/1,1/7,7/6}
                                \path [dashed-edge] (\source) -- (\dest);
                        \foreach \source / \dest in {4/5,5/1,1/8,8/7}
                                \path [dotted-edge] (\source) -- (\dest);
                        \foreach \source / \dest in {5/6,6/1,1/2,2/8}
                                \path [edge] (\source) -- (\dest);
                        \foreach \source / \dest in {1/3}
                                \path [very-thick-edge] (\source) -- (\dest);
    
                    \end{tikzpicture}
                    \caption{The decomposition of $W_8$ into three paths and a single edge. Edge is marked with a bold line and distinct paths are marked with solid, dashed and dotted lines.}
                \end{figure}
        \end{enumerate}
    \end{proof}

\section{Necklaces}\label{section5}
    Necklace is a connected graph having a decomposition into $k \geq 2$ paths, in which all vertices except selected two---$u$, $v$---are different. Vertices $u$, $v$ are starting and ending vertices of all paths. If $k=2$, a necklace is a cycle, so it's possible that a signed graph on such necklace is negative and requires
    $\Delta + 1$ colors to be properly colored. We prove necklaces are $\Delta$-colorable in all other cases.

    \begin{theorem}\label{necklace_theorem}
        Let $G$ be a necklace. If $G$ is not a cycle, then $G$ is of class $1^\pm$.
    \end{theorem}
    \begin{proof}
         $G$ is not a cycle, so $\Delta(G) > 3$. By $D = G_1$, \ldots, $G_{\Delta(G)}$ we denote the decomposition of $G$ into paths. By $u$, $v$ we denote the starting and ending vertices of all the paths from $D$. Without loss of generality, we assume that $m(G_i) \leq m(G_{i+1})$ for $1 \leq i < \Delta(G)$. It follows that if there is a path of length $1$ in $D$, it must be $G_1$. By $u_i^j$ we denote a vertex belonging to path $G_i$ such that the distance between $u_i^j$ and $u$ in $G_i$ is $j$. The definition of $v_i^j$ is analogous. By $S = (G$, $\sigma)$ we denote an arbitrary signed graph with $G$ being its underlying graph. We consider two cases separately:
        \begin{enumerate}
            \item\label{necklace_case_1} $\Delta(G) = 2k + 1$, $k \geq 1$. We will construct necklace $S$ starting with necklace $S_0$ that contains three paths from $D$---$G_1$, $G_2$, $G_3$ and extending it with consecutive pairs of paths. By $S_p$ we denote the necklace constructed from $S_{p-1}$ by adding paths $G_{2p+2}$, $G_{2p+3}$. Clearly, $\Delta(S_0) = 3$. We color edges $u u_2^1$, $v v_3^1$ with color $0$. Let us observe that the remaining edges of $S_0$ span a path that contains following vertices: $u_2^1$, \ldots, $v_2^1$, $v$, $v_1^1$, \ldots, $u_1^1$, $u$, $u_3^1$, \ldots, $v_3^1$. It can be colored with colors $\pm 1$.
            
            Necklace $S_p$ is constructed from $S_{p-1}$ by adding paths $G_{2p+2}$, $G_{2p+3}$. Clearly, $\Delta(S_p) = \Delta(S_{p-1}) + 2$, so there are two new colors $\pm \alpha$ available for coloring of $S_p$. Let us assume there is an edge $u u_l^1$ in $S_{p-1}$ colored with color $0$. We can color edge $u u_{2p+2}^1$ with color $0$ and path $u_l^1$, $u$, $u_{2p+3}^1$, \ldots, $v_{2p+3}^1$, $v$, $v_{2p+2}^1$, \ldots, $u_{2p+2}^1$ with colors $\pm\alpha$. That way edge $u u_l^1$, previously colored with $0$, was recolored. Let us observe that this gives us a $\Delta(S_p)$-edge-coloring of $S_p$. Let us also observe that there is still an edge incident to $u$ such that it's colored with color $0$ in the coloring of $S_p$, so such an edge must exist in the colorings of all $S_0$, \ldots, $S_p$.
            
            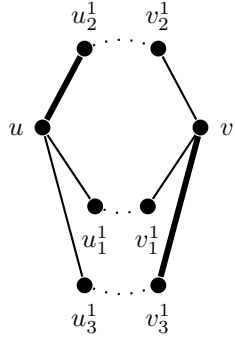
\begin{figure}
                \centering
                \begin{tikzpicture}[thick, scale=0.7]
                    \node[vertex, label={[xshift=-1em, yshift=-1em] $u$}] (1) at (-0.5,0) {};
                    \node[vertex, label={[xshift=1em, yshift=-1em] $v$}] (2) at (2.5,0) {};
                    \node[vertex, label={[xshift=0em, yshift=0em] $u_2^1$}] (3) at (0.3,1.5) {};
                    \node[vertex, label={[xshift=0em, yshift=0em] $v_2^1$}] (4) at (1.7,1.5) {};
                    
                    \node[vertex, label={[xshift=0em, yshift=-2.5em] $u_1^1$}] (5) at (0.5,-1.5) {};
                    \node[vertex, label={[xshift=0em, yshift=-2.5em] $v_1^1$}] (6) at (1.5,-1.5) {};
                    
                    \node[vertex, label={[xshift=0em, yshift=-2.5em] $u_3^1$}] (7) at (0.3,-3) {};
                    \node[vertex, label={[xshift=0em, yshift=-2.5em] $v_3^1$}] (8) at (1.7,-3) {};
                    
                    \foreach \source / \dest in {1/3,8/2}
                            \path [very-thick-edge] (\source) -- (\dest);
                
                    \foreach \source / \dest in {4/2,1/5,6/2,1/7}
                            \path [edge] (\source) -- (\dest);
    
                    \path [possible-path] (3) edge [bend left=20] node {} (4);
                    \path [possible-path] (5) edge [bend right=20] node {} (6);
                    \path [possible-path] (7) edge [bend right=20] node {} (8);

                \end{tikzpicture}
                \caption{Necklace $S_0$ considered in case \ref{necklace_case_1} of the proof of Theorem \ref{necklace_theorem}. Edges colored with color $0$ are marked with bold lines. Remaining edges span a path.}
            \end{figure}
            
            \item\label{necklace_case_2} $\Delta(G) = 2k$, $k \geq 2$. We will construct necklace $S$ starting with necklace $S_0$ that contains four paths from $D$---$G_1$, \ldots, $G_4$ and extending it with consecutive pairs of paths. By $S_p$ we denote the necklace constructed from $S_{p-1}$ by adding paths $G_{2p+3}$, $G_{2p+4}$. Clearly, $\Delta(S_0) = 4$. Let us observe that vertices $u_3^1$, $u$, $u_1^1$, \ldots, $v_1^1$, $v$, $v_2^1$ span a path in $S_0$ and its edges can be colored with colors $\pm 1$. We can assume that incidence $u\colon u u_3^1$ gets color $1$. Let us observe that the remaining edges also span a path, end its edges can be colored with colors $\pm 2$.
            
            Necklace $S_p$ is constructed from $S_{p-1}$ by adding paths $G_{2p+3}$, $G_{2p+4}$. Clearly, $\Delta(S_p) = \Delta(S_{p-1}) + 2$, so there are two new colors $\pm \alpha$ available for coloring of $S_p$. Let us assume there is an incidence $u\colon u u_l^1$ in $S_{p-1}$ colored with color $1$. We can color incidence $u\colon u u_{2p+3}^1$ with color $1$ and $u_{2p+3}^1\colon u u_{2p+3}^1$---with color $-\sigma(u u_{2p+3}^1)$. Path $u_l^1$, $u$, $u_{2p+4}^1$, \ldots, $v_{2p+4}^1$, $v$, $v_{2p+3}^1$, \ldots, $u_{2p+3}^1$ can be colored with colors $\pm \alpha$. That way incidence $u\colon u u_l^1$, previously colored with $1$, was recolored. Let us observe that this gives us a $\Delta(S_p)$-edge-coloring of $S_p$. Let us also observe that there is still an incidence incident to $u$ such that it's colored with color $1$ in the coloring of $S_p$, so such an incidence must exist in the colorings of all $S_0$, \ldots, $S_p$.
            
            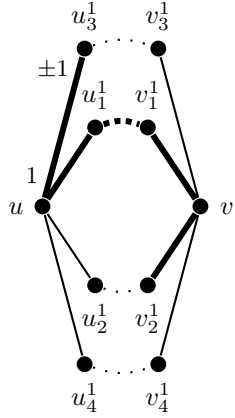
\begin{figure}
                \centering
                \begin{tikzpicture}[thick, scale=0.7]
                    \node[vertex, label={[xshift=-1em, yshift=-1em] $u$}] (1) at (-0.5,0) {};
                    \node[vertex, label={[xshift=1em, yshift=-1em] $v$}] (2) at (2.5,0) {};
                    
                    \node[vertex, label={[xshift=0em, yshift=0em] $u_3^1$}] (9) at (0.3,3) {};
                    \node[vertex, label={[xshift=0em, yshift=0em] $v_3^1$}] (10) at (1.7,3) {};
                    
                    \node[vertex, label={[xshift=0em, yshift=0em] $u_1^1$}] (3) at (0.5,1.5) {};
                    \node[vertex, label={[xshift=0em, yshift=0em] $v_1^1$}] (4) at (1.5,1.5) {};
                    
                    \node[vertex, label={[xshift=0em, yshift=-2.5em] $u_2^1$}] (5) at (0.5,-1.5) {};
                    \node[vertex, label={[xshift=0em, yshift=-2.5em] $v_2^1$}] (6) at (1.5,-1.5) {};
                    
                    \node[vertex, label={[xshift=0em, yshift=-2.5em] $u_4^1$}] (7) at (0.3,-3) {};
                    \node[vertex, label={[xshift=0em, yshift=-2.5em] $v_4^1$}] (8) at (1.7,-3) {};
                    
                    \node[color-label] (uu11) at (-0.7,0.6) {$1$};
                    \node[color-label] (u11u) at (-0.3,2.6) {$\pm 1$};

                    \foreach \source / \dest in {1/9,1/3,2/4,2/6}
                            \path [very-thick-edge] (\source) -- (\dest);
                    
                    \path [very-thick-possible-path] (3) edge [bend left=20] node {} (4);
                
                    \foreach \source / \dest in {1/5,1/7,8/2,10/2}
                            \path [edge] (\source) -- (\dest);

                    \path [possible-path] (9) edge [bend left=20] node {} (10);
                    \path [possible-path] (5) edge [bend right=20] node {} (6);
                    \path [possible-path] (7) edge [bend right=20] node {} (8);

                \end{tikzpicture}
                \caption{Necklace $S_0$ considered in case \ref{necklace_case_2} of the proof of Theorem \ref{necklace_theorem}. Edges colored with colors $\pm 1$ are marked with bold lines. Remaining edges are colored with colors $\pm 2$.}
            \end{figure}
        \end{enumerate}
    \end{proof}

\section{Complete bipartite graphs}\label{section6}
    A graph is bipartite if and only if its vertices can be divided into two sets $V_1$, $V_2$ such that each edge has endpoints in both sets. These sets are usually called \emph{parts}. Graph is a complete bipartite graph if and only if it's bipartite and for every $u \in V_1$, $v \in V_2$ there is an edge $uv$.

    \begin{theorem}\label{bipartite_theorem}
        Let $r$, $t$ be positive integers and $G$ be a complete bipartite graph $K_{r,t}$. If $r\neq t$, then $G$ is of class $1^\pm$.
    \end{theorem}
    \begin{proof}
        Without loss of generality, we assume that $r < t$. Let $t = 2s + k$, $s \in \mathbf{N}$ and $k \in \{0$, $1\}$. Let $V(G) = \{u_0$, \ldots, $u_{r-1}\} \cup \{v_0$, \ldots, $v_{t-1}\}$ and $E(G) = \{u_i v_j \colon 0 \leq i \leq r-1$, $0 \leq j \leq t-1\}$.
        
        Let $0 \leq j \leq s - 1$. Let $G_j$ be a subgraph of $G$ such that $V(G_j) = \{u_0$, \ldots, $u_{r-1}\} \cup \{v_{2j}$, \ldots, $v_{(2j + r) \bmod t}\}$ and $E(G_j) = \{u_i v_{(i + 2j) \bmod t} \colon 0 \leq i \leq r-1\} \cup \{u_i v_{(i + 2j + 1) \bmod t} \colon 0 \leq i \leq r-1\}$. Let us observe that:
        \begin{itemize}
            \item $\deg_{G_j}(u_i) = 2$ for $i = 0$, \ldots, $r - 1$;
            \item $\deg_{G_j}(v_{2j}) = \deg_{G_j}(v_{(2j + r) \bmod t}) = 1$;
            \item $\deg_{G_j}(v_{(2j + i) \bmod t}) = 2$ for $i = 1$, \ldots, $r - 1$.
        \end{itemize}
        It follows from the definition of $G_j$ and above observations that $G_j$ is a path. We will show that any two distinct paths $G_{j_1}$, $G_{j_2}$ are edge disjoint. Let us assume it's not true. Then one of the following cases must hold for some $i_1$, $i_2 \in \{0$, \ldots, $r-1\}$.

        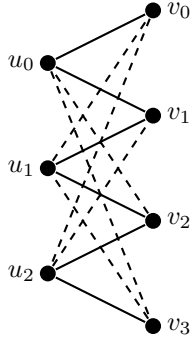
\begin{figure}
            \centering
            \begin{tikzpicture}[thick, scale=0.7]
                \node[vertex, label={[xshift=-1em, yshift=-1em] $u_0$}] (1) at (0,5) {};
                \node[vertex, label={[xshift=-1em, yshift=-1em] $u_1$}] (2) at (0,3) {};
                \node[vertex, label={[xshift=-1em, yshift=-1em] $u_2$}] (3) at (0,1) {};
                
                \node[vertex, label={[xshift=1em, yshift=-1em] $v_0$}] (4) at (2,6) {};
                \node[vertex, label={[xshift=1em, yshift=-1em] $v_1$}] (5) at (2, 4) {};
                \node[vertex, label={[xshift=1em, yshift=-1em] $v_2$}] (6) at (2, 2) {};
                \node[vertex, label={[xshift=1em, yshift=-1em] $v_3$}] (7) at (2, 0) {};
            
                \foreach \source / \dest in {4/1,1/5,5/2,2/6,6/3,3/7}
                        \path [edge] (\source) -- (\dest);
                \foreach \source / \dest in {6/1,1/7,7/2,2/4,4/3,3/5}
                        \path [dashed-edge] (\source) -- (\dest);

            \end{tikzpicture}
            \caption{Paths $G_0$ (solid lines) and $G_1$ (dashed lines) created in the procedure of coloring of complete bipartite graph $K_{3, 4}$ described in the proof of Theorem \ref{bipartite_theorem}.}
        \end{figure}
        
        \begin{enumerate}
            \item $u_{i_1} v_{(i_1 + 2j_1) \bmod t} = u_{i_2} v_{(i_2 + 2j_2) \bmod t}$. Clearly, $u_{i_1} = u_{i_2}$, so $i_1 = i_2$. Moreover, $v_{(i_1 + 2j_1) \bmod t}$ = $v_{(i_2 + 2j_2) \bmod t}$, so $2j_1 \equiv 2j_2 \pmod t$. It follows that $0 \equiv 2 (j_1 - j_2) \pmod t$. Let us consider two cases separately:
            \begin{enumerate}
                \item $t = 2s + 1$. Then $t \equiv 2s + 1 \pmod t$, so $0 \equiv 2s + 1 \pmod t$ and $2s \equiv -1 \pmod t$. We observe that if $0 \equiv 2 (j_1 - j_2) \pmod t$, then $0 \equiv 2s (j_1 - j_2) \pmod t$, so $0 \equiv j_2 - j_1 \pmod t$ and $j_1 \equiv j_2 \pmod t$. It follows that $j_1 = j_2$, a contradiction.
                
                \item $t = 2s$. If $0 \equiv 2 (j_1 - j_2) \pmod t$, then $0 \equiv 2 (j_1 - j_2) \pmod{2s}$ and $0 \equiv j_1 - j_2 \pmod s$, so $j_1 \equiv j_2 \pmod s$. It follows that $j_1 = j_2$, a contradiction.
            \end{enumerate}
        
            \item $u_{i_1} v_{(i_1 + 2j_1 + 1) \bmod t} = u_{i_2} v_{(i_2 + 2j_2 + 1) \bmod t}$. We observe that $i_1 = i_2$, so $2j_1 + 1 \equiv 2j_2 + 1 \pmod t$ and then $2(j_1 - j_2) \equiv 0 \pmod t$. The remaining part of that case is identical to the previous case and leads to a contradiction.
            
            \item Either $u_{i_1} v_{(i_1 + 2j_1) \bmod t} = u_{i_2} v_{(i_2 + 2j_2 + 1) \bmod t}$ or $u_{i_1} v_{(i_1 + 2j_1 + 1) \bmod t} = u_{i_2} v_{(i_2 + 2j_2) \bmod t}$. Without loss of generality, we assume that the first of two holds. Clearly, $i_1 = i_2$, so $2j_1 \equiv 2j_2 + 1 \pmod t$ and then $2(j_1 - j_2) \equiv 1 \pmod t$. We consider two cases separately:
            \begin{enumerate}
                \item $t = 2s + 1$. Then $2s \equiv -1 \pmod t$. Let us observe that $2s(j_1 - j_2) \equiv s \pmod t$, so $j_2 - j_1 \equiv s \pmod t$. It's a contradiction because $j_1$, $j_2 \in \{0$, \ldots, $s-1\}$. 
                \item $t = 2s$. We observe that $2 (j_1 - j_2) - 1 \equiv 0 \pmod{2s}$, so the odd number $2 (j_1 - j_2) - 1$ is divisible by an even number $2s$, a contradiction.
            \end{enumerate}
        \end{enumerate}
        
        All of the cases lead to a contradiction, so paths $G_{j_1}$, $G_{j_2}$ are edge disjoint. Clearly, $\bigcup\limits_{j=0}^{s-1} G_{j} \subseteq G$. We will continue the proof separately for two cases:
        \begin{enumerate}
            \item $t = 2s$. Since $G$ is a complete bipartite graph, $m(G) = rt = 2rs$. We observe that $m(G_j) = 2r$ for $0 \leq j \leq s - 1$. It follows that $m(\bigcup\limits_{j=0}^{s-1} G_{j}) = \sum\limits_{j=0}^{s-1} m(G_{j}) = 2rs$, so $m(G) = m(\bigcup\limits_{j=0}^{s-1} G_{j})$. Since paths $G_j$ are edge disjoint, $G = \bigcup\limits_{j=0}^{s-1} G_{j}$, so $G$ admits a decomposition into exactly $s$ paths. Let us observe that an arbitrary signed graph with $G$ being its underlying graph can be colored with $2s = \Delta(G)$ colors, since each path can be colored with just two colors. It completes a proof for this case.

            \item $t = 2s + 1$. It's clear that $m(G) = rt = r(2s+1)$ and $m(\bigcup\limits_{j=0}^{s-1} G_{j}) = 2rs$. Let $M$ be a subgraph of $G$ such that $V(M) = \{u_0$, \ldots, $u_{r-1}\} \cup \{v_{(i-1) \bmod t} \colon 0 \leq i \leq r - 1\}$ and $E(M) = \{u_i v_{(i-1) \bmod t} \colon 0 \leq i \leq r - 1\}$. It's easy to observe that all the vertices of $M$ have degree equal to $1$, so $M$ is a matching.
        
            We will show that $E(M) \cap E(\bigcup\limits_{j=0}^{s-1} G_{j}) = \emptyset$. We assume that it's not true. Then for some $j$ and $i_1$, $i_2 \in \{0$, \ldots, $r-1\}$ one of the following cases must hold:
            \begin{enumerate}
                \item $u_{i_1} v_{(i_1-1) \bmod t} = u_{i_2} v_{(i_2 + 2j) \bmod t}$. Clearly, $i_1 = i_2$. We observe that $i_1 - 1 \equiv i_2 + 2j \pmod t$, so $2j + 1 \equiv 0 \pmod t$ and then $2js + s \equiv 0 \pmod t$. Since $t = 2s+1$, $2s \equiv -1 \pmod t$, so $j \equiv s \pmod t$. It's a contradiction because $j \in \{0$, \ldots, $s-1\}$.
                \item $u_{i_1} v_{(i_1-1) \bmod t} = u_{i_2} v_{(i_2 + 2j + 1) \bmod t}$. It holds that $i_1 - 1 \equiv i_2 + 2j + 1 \pmod t$, so $2j \equiv -2 \pmod t$ and then $2js \equiv -2s \pmod t$. Since $t = 2s+1$, $j + 1 \equiv 0 \pmod t$. It's a contradiction because $t = 2s + 1$ and $j \in \{0$, \ldots, $s-1\}$.
            \end{enumerate}
            All of the cases lead to a contradiction, so $E(M) \cap E(\bigcup\limits_{j=0}^{s-1} G_{j}) = \emptyset$.
            
            Since $m(M) = r$, $G = E(M) \cup E(\bigcup\limits_{j=0}^{s-1} G_{j})$, $G$ admits a decomposition into exactly $s$ paths and one matching. Let us observe that an arbitrary signed graph with $G$ being its underlying graph can be colored with $2s + 1 = \Delta(G)$ colors, since each path can be colored with two non-zero colors and a matching can be colored with color $0$.
        \end{enumerate}
    \end{proof}

We briefly consider complete bipartite graphs with equal parts. Complete bipartite graph $K_{r, r}$ is an $r$-regular graph, so it follows from Lemma \ref{regular_lemma} that if $r$ is even and $K_{r, r}$ has an odd number of negative edges, it cannot be colored using $\Delta$ colors. We hypothesize that for all other cases, complete bipartite graphs with equal parts are $\Delta$-colorable. If the conjecture is true, graphs with odd $r$ must have a perfect matching $M$ such that $K_{r, r} \setminus M$ have an even number of negative edges and have a decomposition into $(r - 1) / 2$ spanning $2$-regular subgraphs with positive cycles being their components.

\begin{conjecture}
Let $r \in \mathbf{N}$ and $S$ be a signed complete bipartite graph $K_{r, r}$. $\chi'(S) = \Delta(S)$ if one of the following conditions holds:
\begin{enumerate}
    \item $r$ is odd;
    \item $r$ is even and there is an even number of negative edges in $S$.
\end{enumerate}
\end{conjecture}

\bibliography{signed}

\end{document}